\documentclass  [journal,onecolumn,tbtags, draftcls, 12pt]{IEEEtran}
\usepackage{graphicx}

\usepackage{microtype} 

\usepackage[sort,compress]{cite}

\usepackage[cmex10]{mathtools}
\usepackage{dsfont,amssymb}

\usepackage[standard,amsmath,thmmarks] {ntheorem}
\newtheorem {problem}       {Problem}

\let\ALPHABET \mathcal

\let\VEC      \mathbf

\newcommand\IND{\mathds{1}}
\newcommand\PR {\mathds{P}}
\newcommand\EXP{\mathds{E}}

\newcommand\DEFINED{\coloneqq}

\newcommand\reals{\mathds{R}}

\usepackage[T1]{fontenc}

\newcommand\important[1]{{\bfseries\itshape#1}}

\newcommand\COST{l}

\begin{document}

\title {Decentralized Stochastic Control with Partial History Sharing: A Common Information Approach}

\author {Ashutosh Nayyar, Aditya Mahajan and Demosthenis Teneketzis%
        \thanks{Preliminary version of this paper appeared in the proceedings of
        the 46th Allerton conference on communication, control, and computation,
      2008 (see \cite{MahajanNayyarTeneketzis:2008}).}}

\maketitle

\vspace*{-2\baselineskip}
\begin{abstract}
\vspace*{-0.5\baselineskip}
A general model of decentralized stochastic control called \emph{partial history
sharing} information structure is presented. In this model, at each step the
controllers share part of their observation and control history with each other.
This general model subsumes several existing models of information sharing as
special cases. Based on the information commonly known to all the controllers,
the decentralized problem is reformulated as an equivalent centralized problem
from the perspective of a coordinator. The coordinator knows the common
information and select prescriptions that map each controller's local
information to its control actions. The optimal control problem at the
coordinator is shown to be a partially observable Markov decision process
(POMDP) which is solved using techniques from Markov decision theory. This
approach provides (a)~structural results for optimal strategies, and (b)~a
dynamic program for obtaining optimal strategies for all controllers in the
original decentralized problem. Thus, this approach unifies the various ad-hoc
approaches taken in the literature. In addition, the structural results on
optimal control strategies obtained by the proposed approach cannot be obtained
by the existing generic approach (the person-by-person approach) for obtaining
structural results in decentralized problems; and the dynamic program obtained
by the proposed approach is simpler than that obtained by the existing generic
approach (the designer's approach) for obtaining dynamic programs in
decentralized problems. 
\end{abstract}
\vspace*{-1\baselineskip}
\begin{IEEEkeywords}
\vspace*{-0.5\baselineskip}
Decentralized Control, Stochastic Control, Information Structures, Markov Decision Theory, Team Theory 
\end{IEEEkeywords}

\vspace*{-1\baselineskip}
\section{Introduction} \label{sec:introduction}

Stochastic control theory provides analytic and computational techniques for
centralized decision making in stochastic systems with noisy observations.  For
specific models such as Markov decision processes and linear quadratic and
Gaussian systems, stochastic control gives results that are intuitively
appealing and computationally tractable. However, these results are derived
under the assumption that all decisions are made by a centralized decision maker
who sees all observations and perfectly recalls past observations and actions.
This assumption of a centralized decision maker is not true in a number of
modern control applications such as networked control systems, communication and
queuing networks, sensor networks, and smart grids. In such applications,
decisions are made by multiple decision makers who have access to different
information. In this paper, we investigate such problems of \emph{decentralized
stochastic control.} 

The techniques from centralized stochastic control cannot be directly applied to
decentralized control problems. Nonetheless, two general solution approaches that
indirectly use techniques from centralized stochastic control have been used in the
literature: (i)~\emph{the person-by-person approach} which takes the viewpoint of an
individual decision maker (DM); and (ii)~\emph{the designer's approach} which
takes the viewpoint of the collective team of DMs.

The person-by-person approach investigates the decentralized control problem
from the viewpoint of one DM, say DM~$i$ and proceeds as follows:
(i)~arbitrarily fix the strategy of all DMs except DM~$i$; and (ii)~use
centralized stochastic control to derive structural properties for the optimal
best-response strategy of DM~$i$. If such a structural property does not depend
on the choice of the strategy of other DMs, then it also holds for globally
optimal strategy of DM~$i$. By cyclically using this approach for all DMs, we
can identify the \emph{structure} of globally optimal strategies for all DMs.

A variation of this approach may be used to identify
person-by-person optimal strategies. The variation proceeds iteratively as follows.
Start with an initial guess for the strategies of all DMs.  At each iteration,
select one DM (say DM~$i$), and change its strategy to the best
response strategy given the strategy of all other DMs. Repeat the process until
a fixed point is reached, i.e., when no DM can improve performance by
unilaterally changing its strategy. The resulting strategies are
person-by-person optimal \cite{Ho:1980}, and in general,  not globally
optimal. 

In summary, the person-by-person approach identifies structural properties of
globally optimal strategies and provides an iterative method to obtain
person-by-person optimal strategies. This method has been successfully used to
identify structural properties of globally optimal strategies for various
applications including real-time communication~\cite{Witsenhausen:1979,
WalrandVaraiya:1983,Teneketzis:2006, NayyarTeneketzis:2008, KaspiMerhav:2010},
decentralized hypothesis testing and quickest change
detection~\cite{Tenney_Detection, Tsitsiklis_survey, Dec_Wald, VBP_Detection,
VBP_2, NayyarTeneketzis:2009, Nayyar_MTNS,Teneket_Quickest_Det,
V_change_detection}, and networked control systems \cite{WalrandVaraiya:1983a,
MahajanTeneketzis:2009, WuLall:2010}. Under certain conditions, the
person-by-person optimal strategies found by this approach are
globally optimal \cite{Radner:1962, Krainak, Ho:1980}.

The designer's approach, which is developed in~\cite{Witsenhausen:1973, Mahajan_thesis},
investigates the decentralized control problem from the viewpoint of the
collective team of DMs or, equivalently, from the viewpoint of a system designer
who knows the system model and probability distribution of the primitive random
variables and chooses control strategies for all DMs. Effectively, the designer
is solving a centralized planning problem. The designer's approach proceeds by:
(i)~modeling this centralized planning problem as a multi-stage,
\emph{open-loop} stochastic control problem in which the designer's decision at
each time is the control
\emph{law} for that time for all DMs; and (ii)~using centralized stochastic
control to obtain a dynamic programming decomposition. Each step of the
resulting dynamic program is a functional optimization problem (in contrast to
centralized dynamic programming where each step is a parameter optimization
problem).

The designer approach is often used in tandem with the person-by-person approach
as follows. First, the person-by-person approach is used to identify structural
properties of globally optimal strategies. Then, restricting attention to
strategies with the identified structural property, the designer's approach is
used to obtain a dynamic programming decomposition for selecting the globally
optimal strategy. Such a tandem approach has been used in various applications
including real-time communication~\cite{WalrandVaraiya:1983,
MahajanTeneketzis:2007, MahajanTeneketzis:2010}, decentralized hypothesis
testing~\cite{NayyarTeneketzis:2009}, and networked control
systems~\cite{WalrandVaraiya:1983a, MahajanTeneketzis:2009}.

In addition to the above general approaches, other specialized approaches have
been developed to address specific problems in decentralized systems.
Decentralized problems with partially nested information structure were defined
and studied in \cite{HoChu:1972}. Decentralized linear quadratic Gaussian (LQG)
control problems with two controllers and partially nested information structure
were studied in \cite{KimLall:2011,LessardLall:2011}. Partially nested
decentralized LQG problems with controllers connected via a graph were studied
in \cite{Rantzer:2006,Gattami:2009}. A generalization of partial nestedness
called stochastic nestedness was defined and studied in \cite{yuksel:2009}. An
important property of LQG control problems with partially nested information
structure is that there exists an affine control strategy which is globally
optimal. In general, the problem of finding the best affine control strategies
may not be a convex optimization problem. Conditions under which the problem of
determining optimal control strategies within the class of affine control
strategies becomes a convex optimization problem were identified in
\cite{voulgaris:2005,rotkowitz:2006}. 

Decentralized stochastic control problems with specific models of information
sharing among controllers have also been studied in the literature. Examples
include systems with delayed sharing information
structures~\cite{WalrandVaraiya:1978,Aicardi:1987,NMT:2011}, systems with
periodic sharing information structure~\cite{OoiVerboutLudwigWornell:1997},
control sharing information structure~\cite{Bismut:1972,Mahajan:2011}, systems
with broadcast information structure~\cite{WuLall:2010}, and systems with common
and private observations~\cite{MahajanNayyarTeneketzis:2008}.

In this paper, we present a new general model of decentralized stochastic
control called partial history sharing information structure. In this model, we
assume that: (i)~controllers sequentially share part of their past data (past
observations and control) with each other by means of a shared memory; and
(ii)~all controllers have perfect recall of commonly available data (common
information). This model subsumes a large class of decentralized control models
in which information is shared among the controllers.

For this model, we present a general solution methodology that reformulates the
original decentralized problem into an equivalent centralized problem from the
perspective of a coordinator. The coordinator knows the common
information and selects prescriptions that map each controller's local
information to its control actions. The optimal control problem at the
coordinator is shown to be a partially observable Markov decision process
(POMDP) which is solved using techniques from Markov decision theory. This
approach provides (a)~structural results for optimal strategies, and (b)~a
dynamic program for obtaining optimal strategies for all controllers in the
original decentralized problem. Thus, this approach unifies the various ad-hoc
approaches taken in the literature.

A similar solution approach is used in~\cite{NMT:2011} for a model that is a
special case of the model presented in this paper. We present an information
state (Eq.~\eqref{eq:delay-sharing}) for the model of~\cite{NMT:2011} that is simpler than that presented
in~\cite[Theorem~2]{NMT:2011}. A preliminary version of the general solution
approach presented here was presented in~\cite{MahajanNayyarTeneketzis:2008} for
a model that had features (e.g., direct but noisy communication links between
controllers) that are not necessary for partial history sharing. However, it can be shown that by suitable redefinition of variables, the model in \cite{MahajanNayyarTeneketzis:2008} can be recast as an instance of the model in this paper and vice versa (see Appendix \ref{app:equiv}). The information
state for partial history sharing that is presented in this paper (see
Thereom~\ref{thm:newinfostate})  is simpler than that presented
in~\cite[Eq.~(39)]{MahajanNayyarTeneketzis:2008}. 
\subsection {Common Information Approach for a Static Team Problem}

We first illustrate how common information can be used in a static team problem
with two controllers. Let $X$ denote the state of nature and $Y^*$, $Y^1$, $Y^2$
be three correlated random variables that depend on $X$. Assume that the joint
distribution of $(X, Y^*, Y^1, Y^2)$ is given.

Controller~$i$, $i=1,2$, observes $(Y^*, Y^i)$ and chooses a control action $U^i
= g^i(Y^*, Y^i)$. The system incurs a cost $l(X, U^1, U^2)$. The control
objective is to choose $(g^1, g^2)$ to minimize 
\begin{equation*}
  J(g^1, g^2) \DEFINED \EXP^{(g^1, g^2)}[ l(X, U^1, U^2) ] \end{equation*}

If all the system variables are finite valued, we can solve the above
optimization problem by a brute force search over all control strategies $(g^1,
g^2)$. For example, if all variables are binary valued, we need to compute the
performance of $2^4 \times 2^4 = 256$ control strategies and choose the one with
the best performance. 

In this example, both controllers have a common observation $Y^*$. One of the
main ideas of this paper is to use such common information among the controllers
to simplify the search process as follows. Instead of specifying the control
strategies $(g^1, g^2)$ directly, we consider a coordinated system
in which a \emph{coordinator} observes the common information $Y^*$ and chooses
\emph{prescriptions} $(\Gamma^1, \Gamma^2)$ where $\Gamma^i$ is a mapping from
$Y^i$ to $U^i$, $i=1,2$. Hence, $(\Gamma^1, \Gamma^2) = d(Y^*)$, where $d$ is
called the \emph{coordination strategy}. The coordinator then communicates these
prescriptions to the controllers who simply use them to choose $U^i =
\Gamma^i(Y^i)$, $i=1,2$.

It is easy to verify (see Proposition~\ref{prop:equiv} for a formal proof) that
choosing the control strategies $(g^1, g^2)$ in the original
system is equivalent to choosing a coordination strategy $d$ in the coordinated
system. The problem of choosing the best coordination strategy, however, is a
centralized problem in which the coordinator is the only decision-maker. 

For example, consider the case when all system variables are binary valued. 
For any coordination strategy $d$, let $(\gamma_0^1, \gamma_0^2) = d(0)$ and
$(\gamma_1^1, \gamma_1^2) = d(1)$. Then, the cost associated with this
coordination strategy is given as:
\begin{align*}
  J(d) \DEFINED \EXP^{(d)}[ l(X, U^1, U^2) ] 
  &= \mathds{P}(Y^*=0) \EXP[l(X, \gamma_0^1(Y^1), \gamma_0^2(Y^2))|Y^*=0]
  \notag \\
  & \quad +  \mathds{P}(Y^*=1) \EXP[l(X, \gamma_1^1(Y^1), \gamma_1^2(Y^2))|Y^*=1]
\end{align*}
To minimize the above cost, we can minimize the two terms separately. Therefore,
to find the best coordination strategy $d$, we can search for optimal
prescriptions for the cases $Y^* = 0$ and $Y^*=1$ separately. Searching for the best
prescriptions for each of these cases involves computing the performance of $2^2
\times 2^2 = 16$ prescription pairs and choosing the one with the best
performance. Thus, to find the best coordination strategy, we need to evaluate
the performance of $16 + 16 = 32$ prescription pairs. Contrast this with the
$256$ strategies whose costs  we need to evaluate to solve the original problem
by brute force. 

The above example described a static system and illustrates that common
information can be exploited to convert the decentralized optimization problem
into a centralized optimization problem involving a coordinator. 
In this paper, we build upon this basic idea and present a solution approach
based on common information that works for dynamical decentralized systems as
well. Our approach converts the decentralized problem into a centralized stochastic control
problem (in particular, a partially observable Markov decision process),
identifies structure of optimal control strategies, and provides a  dynamic
program like decomposition for the decentralized problem.

\subsection{Contributions of the Paper}

We introduce a general model of  decentralized stochastic control problem in
which multiple controllers share part of their information with each other. We
call this model the \emph{partial history sharing information structure}. This model
subsumes several existing models of  information sharing in decentralized stochastic
control as special cases (see Section~\ref{sec:specialcases}). We establish two
results for our model. Firstly, we establish a structural property of optimal
control strategies. Secondly, we provide   a dynamic programming decomposition of the
problem of finding optimal control strategies. As in
\cite{NMT:2011,MahajanNayyarTeneketzis:2008}, our results are derived using a
common information based approach (see Section~\ref{sec:proof}). This approach
differs from the person-by-person approach and the designer's approach mentioned
earlier. In particular, the structural properties found in this paper cannot be
found by the person-by-person approach described earlier. Moreover, the
dynamic programming decomposition found in this paper is distinct from ---and simpler
than--- the dynamic programming decomposition based on the designer's approach. For a
general framework for using common information in sequential decision making
problems, see \cite{mythesis}.

\subsection{Notation} \label{sec:notation}

Random variables are denoted by upper case letters; their realization by the
corresponding lower case letter. For integers $a \le b$ and $c \le d$, 
$X_{a:b}$ is a short hand for the vector $(X_a,
X_{a+1}, \dots, X_b)$ while $X^{c:d}$ is a short hand for the vector $(X^c,
X^{c+1}, \dots, X^{d})$. When $a > b$, $X_{a:b}$ equals the empty set.
The combined notation $X^{c:d}_{a:b}$ is a short hand
for the vector $(X^j_i : i = a, a+1, \dots, b$, $j = c, c+1, \dots, d)$. In
general, subscripts are used as time index while superscripts are used to index
controllers.  Bold letters $\VEC X$ are used as a short hand  for the vector
$(X^{1:n})$. $\mathds{P}(\cdot)$ is the probability of an event,
$\mathds{E}(\cdot)$ is the expectation of a random variable. For a collection of
functions $\boldsymbol{g}$, we use $\mathds{P}^{\boldsymbol{g}}(\cdot)$ and
$\mathds{E}^{\boldsymbol{g}}(\cdot)$ to denote that the probability
measure/expectation depends on the choice of functions in $\boldsymbol{g}$.
$\IND_A(\cdot)$ is the indicator function of a set $A$. For singleton sets
$\{a\}$, we also denote $\IND_{\{a\}}(\cdot)$ by $\IND_a(\cdot)$.

For a singleton $a$ and a set $B$, $\{a,B\}$ denotes the set $\{a\} \cup B$. For
two sets $A$ and $B$, $\{A, B\}$ denotes the set $A \cup B$. For two finite sets
$\mathcal{A}, \mathcal{B}$,  $F(\mathcal{A},\mathcal{B})$ is the set of all
functions from $\mathcal{A}$ to $\mathcal{B}$. Also, if $\mathcal{A} =
\emptyset$, $F(\mathcal{A},\mathcal{B}):= \mathcal{B}$. For a finite set
$\mathcal{A}$,  $\Delta(\mathcal{A})$ is the set of all probability mass
functions over $\mathcal{A}$. For the ease of exposition, we assume that all
state, observation and control variables take values in finite sets.

For two random variables (or random vectors) $X$ and $Y$ taking values in
$\ALPHABET X$ and $\ALPHABET Y$, $\mathds{P}(X=x|Y)$ denotes the conditional
probability of the event $\{X=x\}$ given $Y$ and $\mathds{P}(X|Y)$ denotes the
conditional PMF (probability mass function) of $X$ given $Y$, that is, it
denotes the collection of conditional probabilities $\mathds{P}(X=x|Y), x \in
\ALPHABET X$. Finally, all equalities involving random variables are to be
interpreted as almost sure equalities (that is, they hold with probability one).

\subsection{Organization} \label{sec:organization}

The rest of this paper is organized as follows. We present our model of a
decentralized stochastic control problem  in Section~\ref{sec:PF}. We also
present several special cases of our model in this section. We prove our main
results in Section~\ref{sec:proof}. We apply our result to some special cases in
Section~\ref{sec:special_results}. We present a simplification of our result and
a generalization of our model in Section~\ref{sec:generalization}. We consider
the infinite time-horizon discounted cost analogue of our problem in
Section~\ref{sec:infinite}. Finally, we conclude in
Section~\ref{sec:conclusion}. 

\section{Problem Formulation} \label{sec:PF}

\subsection{Basic model: Partial History Sharing Information Structure}
\label{sec:model}


\subsubsection{The Dynamic System}

Consider a dynamic system with $n$ controllers. The system operates in discrete
time for a horizon $T$. Let $X_t \in \ALPHABET X_t$ denote the state of the
system at time $t$, $U^i_t \in \ALPHABET U^i_t$ denote the control action of
controller~$i$, $i=1,\dots,n$ at time $t$, and $\VEC U_t$ denote the vector
$(U^1_t, \dots, U^n_t)$.

The initial state $X_1$ has a probability distribution $Q_1$ and evolves according~to
\begin{equation}
  \label{eq:state}
  X_{t+1} = f_t(X_t, \VEC U_t, W^0_t),
\end{equation}
where $\{W^0_t\}_{t=1}^{T}$ is a sequence of i.i.d\@. random
variables with probability distribution~$Q^0_W$.

\subsubsection{Data available at the controller}

At any time $t$, each controller has access to three types of data: current
observation, local memory, and shared memory.
\begin{enumerate}
  \item[(i)]  \important{Current local observation}: 
    Each controller makes a local observation $Y^i_t \in \ALPHABET Y^i_t$ on the state of the system at time $t$,
    \begin{equation}  \label{eq:observation}
      Y^i_t = h^i_t(X_t, W^i_t),
    \end{equation}
    where $\{W^i_t\}_{t=1}^{T}$ is a sequence of i.i.d\@. random variables
    with probability distribution $Q^i_W$. 
    We assume that the random variables
    in the collection $\{X_1,W^j_t,t=1,\dots,T, j =0, 1,\dots, n\}$, called \emph{primitive random variables},  are 
 mutually  independent. 

  \item [(ii)] \important{Local memory  }:
            Each controller stores a subset $M^i_t$ of its past local observations and its past actions in a local memory:
        \begin{equation}\label{eq:memory}
          M^i_t \subset \{Y^i_{1:t-1},U^i_{1:t-1}\}.
        \end{equation}
        At $t=1$, the local memory is empty, $M^i_1 = \emptyset$.

      \item [(iii)] \important{Shared memory}:
        In addition to its local memory, each controller has access to a shared
        memory. The contents $C_t$ of the shared memory at time $t$ are a
        subset of the past local observations and control actions of all
        controllers:
        \begin{equation}\label{eq:shared memory}
          C_t \subset \{\VEC Y_{1:t-1}, \VEC U_{1:t-1}\}
        \end{equation}
        where $\VEC Y_t$ and $\VEC U_t$ denote the vectors $(Y^1_t, \dots, Y^n_t)$ and 
$(U^1_t, \dots, U^n_t)$ respectively. At $t=1$, the shared memory is empty, $C_1 = \emptyset$.
    
\end{enumerate}

Controller~$i$ chooses action $U^i_t$ as a function of the total data $(Y^i_t,
M^i_t, C_t)$ available to it.  Specifically, for every controller~$i$, $i = 1,
\dots, n$, 
\begin{equation}
  \label{eq:control}
  U^i_t = g^i_t(Y^i_t, M^i_t, C_t),
\end{equation}
where $g^i_t$ is called the \emph{control law} of controller~$i$. The collection
$\VEC g^i = (g^i_1, \dots, g^i_T)$ is called the \emph{control strategy} of
controller~$i$. The collection $\VEC g^{1:n} = (\VEC g^1, \dots, \VEC g^n)$ is
called the \emph{control strategy} of the system.

\subsubsection{Update of local and shared memories}

\begin{enumerate}
  \item [(i)] \emph{Shared memory update}: 
    After taking the control action at time $t$, the local information at
    controller~$i$ consists of the contents $M^i_t$ of its local memory, its
    local observation $Y^i_t$ and its control action $U^i_t$. Controller~$i$
    sends a subset $Z^i_t$ of this local information $\{M^i_t, Y^i_t, U^i_t\}$
    to the shared memory. The subset $Z^i_t$ is chosen according to a
    pre-specified protocol. The contents of shared memory are nested in time, that is, the contents $C_{t+1}$
    of the shared memory at time $t+1$ are the contents $C_t$ at time $t$ augmented
    with the new data $\VEC Z_t = (Z^1_t,Z^2_t,\ldots,Z^n_t)$ sent by all the controllers at time $t$:
    \begin{equation} \label{eq:shared_update}
      C_{t+1} = \{C_t, \VEC Z_t\}.
    \end{equation}

  \item [(ii)] \emph{Local memory update}:
    After taking the control action and sending data to the shared memory at
    time $t$, controller~$i$ updates its local memory according to a
    pre-specified protocol. The content $M^i_{t+1}$ of the local memory can at
    most equal the total local information $\{M^i_t, Y^i_t, U^i_t\}$ at the
    controller. However, to ensure that the local and shared memories at time
    $t+1$ don't overlap, we assume that 
    \begin{equation} \label{eq:local update}
      M^i_{t+1} \subset \{M^i_t,Y^i_t,U^i_t\}\setminus Z^i_t.
    \end{equation}
\end{enumerate}

Figure 1 shows the time order of observations, actions and memory updates.
\begin{figure}[htp]
\begin{center}
\includegraphics[width=7.5cm]{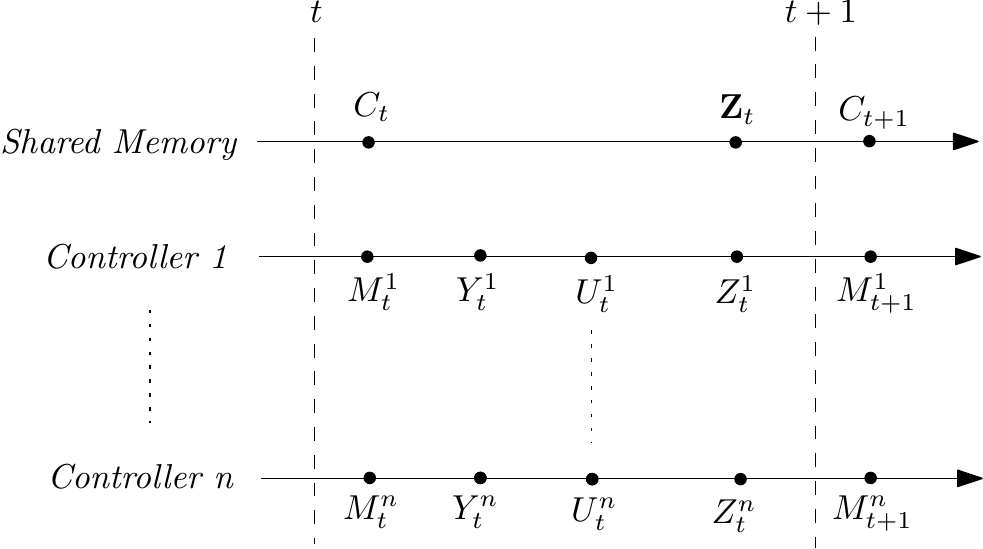}
\caption{Time ordering of Observations, Actions and Memory Updates}
\end{center}
\end{figure}
We refer to the above model as the \emph{partial history sharing information structure}.

\subsubsection{The optimization problem}


At time $t$, the system incurs a cost $\COST(X_t, \VEC U_t)$. The performance of
the control strategy of the system is measured by the expected total cost
\begin{equation}
  \label{eq:cost}
  J(\VEC g^{1:n}) \DEFINED \EXP^{\VEC g^{1:n}}\Big[ 
  \sum_{t=1}^T \COST(X_t, \VEC U_t) \Big],
\end{equation}
where the expectation is with respect to the joint probability measure on
$(X_{1:T}, \VEC U_{1:T})$ induced by the choice of $\VEC g^{1:n}$. 

We are interested in the following optimization problem.
\begin{problem}
  \label{prob:finite}
  For the model described above, given the state evolution functions $f_t$, the
  observation functions $h^i_t$, the protocols for updating local and share
  memory,  the cost function $l$, the distributions $Q_1$, $Q^i_W$,
  $i=0,1,\dots,n$, and the horizon $T$, find a control strategy $\VEC g^{1:n}$
  for the system that minimizes the expected total  cost given
  by~\eqref{eq:cost}.
\end{problem}

\subsection{Special Cases: The Models} \label{sec:specialcases}

In the above model, although we have not specified the exact protocols by which
controllers update the local and shared memories, we assume that  pre-specified
protocols are being used. Different choices of this protocol result in different
information structures for the system. In this section, we describe several
models of decentralized control systems that can be viewed as special cases of
our model by assuming a particular choice of protocol for local and shared
memory updates.

\subsubsection{Delayed Sharing Information Structure} \label{sec:ex1}

Consider the following special case of the model of Section~\ref{sec:model}. 

\begin{enumerate}
  \item[(i)] The shared memory at the beginning of time $t$ is $C_t = \{\VEC Y_{1:t-s},\VEC U_{1:t-s}\}$,  where $s \geq 1$ is a fixed number. The local memory at the beginning of time $t$ is $M^i_t = \{Y^i_{t-s+1:t-1},U^i_{t-s+1:t-1}\}$.
  \item [(ii)]  At each time $t$, after taking the action $U^i_t$, controller~$i$ sends
    $Z^i_t = \{Y^i_{t-s+1},U^i_{t-s+1}\}$ to the shared memory and the shared memory at $t+1$ becomes $C_{t+1} = \{\VEC Y_{1:t-s+1},\VEC U_{1:t-s+1}\}$.
    
  \item[(iii)]  After sending $Z^i_t=\{Y^i_{t-s+1},U^i_{t-s+1}\}$ to the shared memory, controller~$i$ updates the local memory to
    $M^i_{t+1} = \{Y^i_{t-s+2:t},U^i_{t-s+2:t}\}$. 
\end{enumerate}

In this spacial case, the observations and control actions of each controller
are shared with every other controller after a delay of $s$ time steps. Hence,
the above special case corresponds to the delayed sharing information structure
considered in \cite{Witsenhausen:1971, WalrandVaraiya:1978, NMT:2011}.

\subsubsection{Delayed State Sharing Information Structure} \label{sec:ex2}

A special case of the delayed sharing information structure (which itself is a
special case of our basic model) is the \emph{delayed state sharing} information
structure \cite{Aicardi:1987}. This information structure can be obtained from
the delayed sharing information structure by making the following assumptions:
\begin{enumerate}
  \item[(i)] The state of the system at time $t$ is a $n$-dimensional vector
    $X_t = (X^1_t,X^2_t,\ldots,X^n_t)$.
  \item[(ii)] At each time $t$, the current local observation of controller $i$
    is \( Y^i_t = X^i_t, \) for $i =1,2,\ldots,n$.
\end{enumerate} 
In this spacial case, the complete state vector $X_t$ is available to all
controllers after a delay of $s$ time steps.

\subsubsection{Periodic Sharing Information Structure} \label{sec:ex3}

Consider the following special case of the model of Section~\ref{sec:model}
where controllers update the shared memory periodically with period $s \geq 1$:
\begin{enumerate}
  \item [(i)]  For time $ ks < t \leq (k+1)s$, where $k=0,1,2,\ldots$, the
    shared memory at the beginning of time $t$ is
    \(
      C_t = \{\VEC Y_{1:ks}, \VEC U_{1:ks}\}.
    \)
    The local memory at the beginning of time $t$  is
    \(
      M^i_t = \{Y^i_{ks+1:t-1},U^i_{ks+1:t-1}\}.
    \)
  \item[(ii)]
    At each time $t = (k+1)s$, $k=1,2,\dots$, after taking the action $U^i_t$, controller~$i$ sends
    $Z^i_t = \{Y^i_{ks+1:(k+1)s},U^i_{ks+1:(k+1)s}\}$ to the shared memory. At
    other times, each controller does not send anything (thus $Z^i_t =
    \emptyset$). 
   
  \item [(iii)] After sending $Z^i_t$ to the shared memory, controller~$i$
    updates the local memory to $M^i_{t+1} = \{M^i_t,Y^i_t,U^i_t\}\setminus
    Z^i_t$.
\end{enumerate}
In this spacial case, the entire history of observations and control actions are
shared periodically between controllers with period $s$. Hence, the above
special case corresponds to the periodic sharing information structure
considered in \cite{OoiVerboutLudwigWornell:1997}.

\subsubsection{Control Sharing Information Structure} \label{sec:ex4}

Consider the following special case of the model of Section~\ref{sec:model}. 

\begin{enumerate}
  \item [(i)] The shared memory at the beginning of time $t$ is $C_t = \{\VEC
      U_{1:t-1}\}$. The local memory at the beginning of time $t$ is $M^i_t = \{
        Y^i_{1:t-1}\}$.

  \item [(ii)] At each time $t$, after taking the action $U^i_t$, controller~$i$
    sends $Z^i_t = \{U^i_{t}\}$ to the shared memory.

  \item[(iii)]  After sending $Z^i_t = U^i_t$ to the shared memory,
    controller~$i$ updates the local memory to $M^i_{t+1} = Y^i_{1:t}$. 
\end{enumerate}
In this spacial case, the control actions of each controller are shared with
every other controller after a delay of $1$ time step. Hence, the above special
case corresponds to the control sharing information structure considered in
\cite{Bismut:1972}. 

\subsubsection{No Shared Memory with or without finite local memory}  \label{sec:ex5}

Consider the following special case of the model of Section~\ref{sec:model}. 

\begin{enumerate}
  \item [(i)] The shared memory at each time is empty, $C_t = \emptyset$ and the
    local memory at the beginning of time $t$ is $M^i_t = \{Y^i_{t-s:t-1},
    U^i_{t-s:t-1}\}$, where $s\geq1$ is a fixed number.

  \item[(ii)] Controllers do not send any data to shared memory, $Z^i_t =
    \emptyset$.
  \item[(iii)] At the end of time $t$, controllers update their local memories
    to $M^i_{t+1} = \{Y^i_{t-s+1:t}, U^i_{t-s+1:t}\}$.
\end{enumerate}
In this special case, the controllers don't share any data. The above model is
related to the finite-memory controller model of \cite{Sandell:phd}. A related
special case is the situation where the local memory at each controller consists
of all of its past local observations and its past actions, that is, $M^i_t =
\{Y^i_{1:t-1}, U^i_{1:t-1}\}$. 

\begin{remark}
   All the special cases considered above are examples of \emph{symmetric
   sharing}. That is, different controllers  update their local memories
   according to identical protocols  and the data sent by a controller to the
   shared memory is selected according to identical protocols. However, this
   symmetry is not required for our model. Consider for example, the delayed
   sharing information structure where at the end of time $t$, controller~$i$
   sends $Y^i_{t-s_i},U^i_{t-s_i}$ to the shared memory, with $s_i,
   i=1,2,\ldots,n,$  being fixed, but not necessarily identical, numbers. This
   kind of \emph{asymmetric sharing} is also a special case of our model. 
\end{remark}
   
\section{Main Results} \label{sec:proof}

For centralized systems, stochastic control theory provides two important
analytical results. Firstly, it provides a \emph{structural result}. This result
states that there is an optimal control strategy which selects control actions
as a function only of the controller's posterior belief on the state of the
system conditioned on all its observations and actions till the current time.
The controller's posterior belief is called its \emph{information state}.
Secondly, stochastic control theory provides a \emph{dynamic programming decomposition}
of the problem of finding optimal control strategies in centralized systems.
This dynamic programming decomposition 
allows one to evaluate the optimal action for each realization of the
controller's information state in a backward inductive manner.

In this section, we provide a structural result and a dynamic programming decomposition
for the decentralized stochastic control problem with partial information
sharing formulated above (Problem~\ref{prob:finite}). The main idea of the proof
is  to formulate an equivalent centralized stochastic control problem; solve the
equivalent problem using classical stochastic-control techniques; and translate
the results back to the basic model. For that matter, we proceed as follows:
\begin{enumerate}
  \item Formulate a centralized \emph{coordinated system} from the point   of
    view of a \emph{coordinator} that observes only the common information among
    the controllers in the basic model, i.e., the coordinator observes the
    shared memory $C_t$  but not the local memories $(M^i_t$, $i=1,\dots,n)$ or
    local observations $(Y^i_t$, $i=1,\dots,n)$. The coordinator chooses
    prescriptions $\mathbf \Gamma_t = (\Gamma^1_t, \dots, \Gamma^n_t)$, where
    $\Gamma^i_t$ is a mapping from $(Y^i_t, M^i_t)$ to $U^i_t$, $i=1,\dots,n$.

  \item Show that the coordinated system is a POMDP (partially     observable
    Markov decision process).

  \item For the coordinated system, determine the structure of an optimal
    coordination strategy and a dynamic program to find an optimal
    coordination strategy.

  \item Show that any strategy of the coordinated system is implementable in the
    basic model with the same value of the total expected cost. Conversely, any
    strategy of the basic model is implementable in the coordinated system with
    the same value of the total expected cost. Hence, the two systems are
    equivalent.

  \item Translate the structural results and dynamic programming decomposition
    of the coordinated system (obtained in stage~3) to the basic model.
\end{enumerate}

\subsection* {\textbf{Stage 1: The coordinated system}}

Consider a \emph{coordinated system} that consists of a coordinator and $n$
passive controllers. The coordinator knows the shared memory $C_t$ at time~$t$,
but not the local memories $(M^i_t$, $i=1,\dots,n)$ or local observations
$(Y^i_t$, $i=1,\dots,n)$. At each time $t$, the coordinator  chooses mappings
$\Gamma^i_t : \mathcal{Y}^i_t \times \mathcal{M}^i_t \mapsto \mathcal{U}^i_t$,
$i=1,2,\ldots,n$, according to
\begin{equation} \label{eq:coordinator}
  \mathbf{\Gamma_t} = d_t(C_t,\mathbf{\Gamma_{1:t-1}}),
\end{equation} 
where $\mathbf{\Gamma_t} = (\Gamma^1_t,\Gamma^2_t,\ldots,\Gamma^n_t)$. The
function $d_t$ is called the \emph{coordination rule} at time $t$ and the
collection of functions $\VEC d \DEFINED (d_1, \dots, d_T)$ is called the
\emph{coordination   strategy}. The selected $\Gamma^i_t$ is communicated to
controller $i$ at time $t$.

The function $\Gamma^i_t$ tells controller~$i$  how to process its current local
observation and its local memory at time $t$; for that reason, we call
$\Gamma^i_t$ the \emph{coordinator's prescription} to controller~$i$.
Controller~$i$ generates an action using its prescription as follows:
\begin{equation} \label{eq:prescription}
  U^i_t = \Gamma^i_t(Y^i_t,M^i_t).
\end{equation}

For this coordinated system, the system dynamics, the observation model and the
cost are the same as the basic model of Section~\ref{sec:model}: the system
dynamics are given by~\eqref{eq:state}, each controller's current observation is
given by~\eqref{eq:observation} and the instantaneous cost at time~$t$ is
$l(X_t,\mathbf{U}_t)$. As before, the performance of a coordination strategy is
measured by the expected total cost
\begin{equation}
  \label{eq:coordination-cost}
  \hat J(\VEC d)  = \EXP\Big[ \sum_{t=1}^T \COST(X_t, \mathbf{U}_t) \Big],
\end{equation}
where the expectation is with respect to a joint measure on $(X_{1:T},
\mathbf{U}_{1:T})$ induced by the choice of~$\VEC d$. 

In this coordinated system, we are interested in the following optimization
problem:
\begin{problem}
  \label{prob:coordinator}
  For the model of the coordinated system described above, find a coordination
  strategy $\VEC d$ that minimizes the total expected cost given
  by~\eqref{eq:coordination-cost}.
\end{problem}

\subsection*{\textbf{Stage 2: The coordinated system as a POMDP}}

We will now show that the coordinated system is a partially observed Markov
decision process. For that matter, we first describe the model of POMDPs
\cite{Whittle:1983}.

\subsubsection*{POMDP Model}
A partially observable Markov decision process consists of a state process $S_t
\in \mathcal{S}$, an observation process $O_t \in \mathcal{O}$, an action
process $A_t \in \mathcal{A}$, $t=1,2,\ldots,T$, and a single decision-maker
where
\begin{enumerate}
  \item The action at time $t$ is chosen by the decision-maker as a function of
    observation and action history, that is,
    \begin{equation}
      A_t = d_t(O_{1:t}, A_{1:t-1}),
    \end{equation}
    $d_t$ is the decision rule at time $t$.
  \item After the action at time $t$ is taken, the new state and new observation
    are generated according to the transition probability rule 
    \begin{equation}
      \mathds{P}(S_{t+1},O_{t+1}|S_{1:t},O_{1:t},A_{1:t}) = \mathds{P}(S_{t+1},O_{t+1}|S_t,A_t).
    \end{equation}
  \item At each time, an instantaneous cost $\tilde \COST(S_t,A_t)$ is incurred.
  \item The  optimization problem for the decision-maker is to choose a decision
    strategy $\VEC d := (d_1,\ldots,d_T)$ to minimize a total cost given as
    \begin{equation}
      \mathds{E}[\sum_{t=1}^T \tilde l(S_t,A_t)].
    \end{equation}
\end{enumerate}

The following well-known results provides the structure of optimal strategies
and a dynamic program for POMDPs. For details, see \cite{Whittle:1983}.
\begin{theorem}[POMDP Result] \label{thm:pomdp} Let $\Theta_t$ be the
  conditional probability distribution of the state $S_t$ at time $t$ given the
  observations $O_{1:t}$ and actions $A_{1:t-1}$,
  \[ \Theta_t(s) = \mathds{P}(S_t=s|O_{1:t},A_{1:t-1}), ~~~s \in \mathcal{S}.\]
  Then,
  \begin{enumerate}
    \item $\Theta_{t+1} = \eta_t(\Theta_t,A_t,O_{t+1})$, where $\eta_t$ is the
      standard non-linear filter: If $\theta_t,a_t,o_{t+1}$ are the realizations
      of $\Theta_t,A_t$ and $O_{t+1}$, then the realization of $s^{th}$ element
      of the vector $\Theta_{t+1}$ is
      \begin{align}
        \theta_{t+1}(s) &= \frac{\sum_{s'}\theta_t(s')\mathds{P}(S_{t+1}=s,O_{t+1}=o_{t+1}|S_t=s',A_t=a_t)}{\sum_{\hat{s},\tilde{s}}\theta_t(\hat{s})\mathds{P}(S_{t+1}=\tilde{s},O_{t+1}=o_{t+1}|S_t=\hat{s},A_t=a_t)} \nonumber \\
        &=: \eta^s_t(\theta_t,a_t,o_{t+1})
      \end{align}
      and $\eta_t(\theta_t,a_t,o_{t+1})$ is the vector
      $(\eta^s_t(\theta_t,a_t,o_{t+1}))_{s \in \mathcal{S}}$.

    \item There exists an optimal decision strategy of the form
      \[ A_t = \hat{d}_t(\Theta_t). \]
      Further, such a strategy can be found by the following dynamic program:
      \begin{equation}
        V_{T}(\theta) = \inf _{a} 
        \mathds{E}\{\tilde l(S_T,a) |  \Theta_T=\theta\},
      \end{equation}
      and for $1 \leq t \leq T-1$,
      \begin{equation}
        V_{t}(\theta) = \inf _{a} 
        \mathds{E}\big\{\tilde l(S_t,a)+ V_{t+1}(\eta_t(\theta,a,O_{t+1})) \big| \Theta_t=\theta, A_t=a \big\}.
      \end{equation}
  \end{enumerate}
\end{theorem}

We will now show that the coordinated system can be viewed as an instance of the
above POMDP model by defining the state process as  \(S_t := \{X_t, \VEC Y_t,
\VEC M_t \},\) the observation process as \(O_{t} := \VEC Z_{t-1},\) and the
action process \(A_t := \VEC \Gamma_t.\)

\begin{lemma} \label{lemma:state}
  For the coordinated system of Problem \ref{prob:coordinator},
  \begin{enumerate}
    \item   There exist functions $\tilde f_t$ and $\tilde h_t$, $t=1,\dots,T$,
      such that
      \begin{gather} \label{eq:coord_state}
        S_{t+1} = \tilde f_t(S_t, \VEC \Gamma_t, W^0_t,  \VEC W_{t+1}), \\
        \shortintertext{and}
        \VEC Z_{t} = \tilde h_{t}(S_t, \VEC \Gamma_t).
      \end{gather}
      In particular, we have that 
      \begin{equation} \label{eq:markov prop}
        \mathds{P}(S_{t+1},\VEC Z_t|S_{1:t},\VEC Z_{1:t-1},\VEC \Gamma_{1:t}) = \mathds{P}(S_{t+1}, \VEC Z_{t}|S_t,\VEC \Gamma_t).
      \end{equation}

    \item  Furthermore, there exists a function $\tilde \COST$ such that
      \begin{equation}
        \COST(X_t, \VEC U_t) = 
        \tilde \COST(S_t, \VEC \Gamma_t).
      \end{equation}
      Thus, the objective of
      minimizing~\eqref{eq:coordination-cost} is same as minimizing
      \begin{equation}
        \label{eq:coordination-cost-2}
        \hat J(\VEC d)  = \EXP\Big[ \sum_{t=1}^T \tilde \COST(S_t, \VEC \Gamma_t) \Big].
      \end{equation}

  \end{enumerate}
\end{lemma}

\begin{proof}
  The existence of $\tilde f_t$ follows from~\eqref{eq:state},
  \eqref{eq:observation}, \eqref{eq:prescription}, \eqref{eq:local update} and
  the definition of $S_t$. The existence of $\tilde h_t$ follows from the fact
  that $Z^i_t$ is a fixed subset of $\{M^i_t,Y^i_t,U^i_t\}$,
  equation~\eqref{eq:prescription} and the definition of $S_t$. Equation
  \eqref{eq:markov prop} follows from \eqref{eq:coord_state} and the
  independence of $W^0_t, \VEC W_{t+1}$ from all random variables in the
  conditioning in the left hand side of \eqref{eq:markov prop}. The existence of
  $\tilde \COST$ follows from the definition of $S_t$
  and~\eqref{eq:prescription}.
\end{proof}

Recall that the coordinator is choosing its actions according to a coordination
strategy of the form
\begin{equation} \label{eq:coordinator2}
  \mathbf{\Gamma_t} = d_t(C_t,\mathbf{\Gamma_{1:t-1}}) = d_t(\VEC Z_{1:t-1},\mathbf{\Gamma_{1:t-1}}).
\end{equation}
Equation \eqref{eq:coordinator2} and Lemma~\ref{lemma:state} imply that the
coordinated system is an instance of the POMDP model described  above. 

\subsection* {\textbf{Stage 3: Structural result and dynamic program for the coordinated system}}

Since the coordinated system is a POMDP, Theorem~\ref{thm:pomdp} gives the
structure of the optimal coordination strategies. For that matter, define
coordinator's information state 
\begin{equation}
  \Pi_t := \PR(S_t \mid \VEC Z_{1:t-1},
  \VEC \Gamma_{1:t-1}) = \PR(S_t \mid C_t,
  \VEC \Gamma_{1:t-1}) .
\end{equation}
Then, we have the following:
\begin{proposition} \label{prop:structure}
  For Problem~\ref{prob:coordinator}, there is no loss of optimality in
  restricting attention to coordination rules of the form
  \begin{equation}
    \VEC \Gamma_t = \hat{d}_t(\Pi_t).
  \end{equation}
\end{proposition}

Furthermore, an optimal coordination strategy of the above form can be found using
a dynamic program. For that matter, observe that we can write
\begin{equation} \label{eq:updatecoordinator}
  \Pi_{t+1} = \eta_t(\Pi_t, \VEC Z_{t}, \VEC \Gamma_t)
\end{equation}
where $\eta_t$ is the standard non-linear filtering update function (see
Appendix~\ref{sec:update_func}). We denote by $\mathcal{B}_t$ the space of
possible realizations of $\Pi_t$. Thus,
\begin{equation} \label{eq:beliefspace}
  \mathcal{B}_t := \Delta(\mathcal{X}_t \times \mathcal{Y}^1_t \times \mathcal{M}^1_t \times \ldots \times \mathcal{Y}^n_t \times \mathcal{M}^n_t).
\end{equation} 
Recall that $F(\mathcal{Y}^i_t \times \mathcal{M}^i_t,\mathcal{U}^i_t)$ is the
set of all functions from $\mathcal{Y}^i_t \times \mathcal{M}^i_t$ to
$\mathcal{U}^i_t$ (see Section \ref{sec:notation}). Then, we have the following
result.

\begin{proposition} \label{prop:DP}
  For all $\pi_t$ in $\mathcal{B}_t$, define
  \begin{equation}
    V_{T}(\pi) = \inf _{\{\tilde\gamma^i_T \in F(\mathcal{Y}^i_T \times \mathcal{M}^i_T,\mathcal{U}^i_T), 1\leq i \leq n\}} 
    \EXP[ \tilde \COST(S_t, \VEC \Gamma_T) \mid 
      \Pi_t = \pi, \VEC \Gamma_T = (\gamma^1_T,\ldots,\gamma^n_T) ],
    \end{equation}
    and for $1 \leq t \leq T-1$,
    \begin{multline}
      V_{t}(\pi) = \inf _{\{\tilde\gamma^i \in F(\mathcal{Y}^i_t \times \mathcal{M}^i_t,\mathcal{U}^i_t), 1\leq i\leq n\}} 
      \EXP[ \tilde \COST(S_t, \VEC \Gamma_t) + V_{t+1}(\eta_t(\Pi_t, \VEC \Gamma_t, \VEC Z_t) \mid 
        \Pi_t = \pi, \VEC \Gamma_t = (\gamma^1_t,\ldots,\gamma^n_t)].
      \end{multline}

      Then the $\arg \inf$ at each time step gives the coordinator's optimal prescriptions for the controllers when the coordinator's information state is $\pi$.
\end{proposition}

Proposition~\ref{prop:DP} gives a dynamic program for the coordinator's problem
(Problem~\ref{prob:coordinator}). Since the coordinated system is a POMDP, it
implies that computational algorithms for POMDPs can be used to solve the
dynamic program for the coordinator's problem as well. We refer the reader to
\cite{Zhang:2009} and references therein for a review of algorithms to solve POMDPs. 


\subsection*{\textbf{Stage 4: Equivalence between the two models}}

We first observe that since $C_s \subset C_t$, for all $s < t$,  under any given
coordination strategy $\VEC d$, we can use $C_t$ to evaluate the past
prescriptions by recursive substitution.
For example, for $t=2,3$, the past prescriptions can be evaluated as functions
of $C_2$, $C_3$ as follows:
\[\VEC \Gamma_1 = d_1(C_1)  =: \tilde d_1(C_2), \] 
\[\VEC \Gamma_2 = d_2(C_2, \VEC \Gamma_1) = d_2(C_2,\tilde d_1(C_2)) =: \tilde d_2(C_3)\] 
We can now state the following result.

\begin{proposition} \label{prop:equiv}
  The basic model of Section \ref{sec:model} and the coordinated system are equivalent. More
  precisely:
  \begin{enumerate}
    \item[(a)] Given any control strategy $\VEC g^{1:n}$ for the basic model,
      choose a coordination strategy $\VEC d$ for the coordinated system of
      stage~1 as
      \[ d_t(C_t) =
          \big(g^1_t( \cdot, \cdot, C_t), \dots, g^n_t( \cdot, \cdot, C_t)
\big). \]
      Then $\hat J(\VEC d) = J(\VEC g^{1:n})$.

    \item[(b)] Conversely, for any coordination strategy for the coordinated
      system, choose a control strategy $\VEC g^{1:n}$ for the basic model as
      \begin{gather*}
        g^i_1(\cdot,\cdot, C_1) = d^i_1(C_1), \\
      \shortintertext{and}
        g^i_t(\cdot,\cdot, C_t) = d^i_t(C_t, \VEC \Gamma_{1:t-1}), 
      \end{gather*}
      where $\VEC \Gamma_k = d_k(C_k,\VEC \Gamma_{1:k-1})$, $k=1,2,\ldots,t-1$
      and $d^i_t(\cdot)$ is the $i$-th component of $d_t(\cdot)$ (that is,
      $d^i_t(\cdot)$ gives the coordinator's prescription for the $i$-th
      controller). Then, $J(\VEC g^{1:n}) = \hat J(\VEC d)$.
  \end{enumerate}
\end{proposition}

\begin{proof}
See Appendix \ref{sec:equiv_proof}.
\end{proof}

\subsection*{\textbf{Stage 5: Structural result and dynamic program for the basic model}}

Combining Proposition~\ref{prop:structure}  with Proposition~\ref{prop:equiv},
we get the following structural result for Problem~\ref{prob:finite}.

\begin{theorem}[Structural Result for Optimal Control Strategies] \label{thm:str_result}
  In Problem~\ref{prob:finite},  there exist optimal  control strategies of the form
  \begin{equation} \label{eq:our_result}
    U^i_t = \hat{g}^i_t(Y^i_t, M^i_t, \Pi_t), \quad i= 1,2,\ldots,n,
  \end{equation}
  where $\Pi_t$ is the conditional distribution on $X_t, \VEC Y_t, \VEC M_t$ given $C_t$, defined as  
  \begin{equation} \label{eq:define_pi}
    \Pi_t(x, \VEC y, \VEC m) :=    \mathds{P}^{\hat{g}^{1:n}_{1:t-1}}(X_{t}=x,
    \VEC Y_t=\VEC y,\VEC M_t=\VEC m | C_t),
  \end{equation}    
  for all possible realizations $(x, \VEC y, \VEC m)$ of ($X_t, \VEC Y_t,\VEC M_t$).
\end{theorem}

We call $\Pi_t$ the \emph{common information state}. Recall that $\Pi_t$ takes
values in the set $\mathcal{B}_t$ defined in \eqref{eq:beliefspace}.
 
Consider a control strategy $\mathbf{\hat{g}^i}$ for controller $i$ of the form
specified in Theorem~\ref{thm:str_result}. The control law $\hat{g}^i_t$ at time
$t$ is a function from the space $\mathcal{Y}^i_t \times \mathcal{M}^i_t \times
\mathcal{B}_t$  to the space of decisions $\mathcal{U}^i_t$. Equivalently, the
control law $\hat{g}^i_t$ can be represented as a collection of functions
$\{\hat{g}^i_t(\cdot,\cdot,\pi)\}_{\pi \in \mathcal{B}_t}$, where each element
of this collection is a function from $\mathcal{Y}^i_t \times \mathcal{M}^i_t $
to  $\mathcal{U}^i_t$. An element $\hat{g}^i_t(\cdot, \cdot,\pi)$ of this
collection specifies a control action for each possible realization of $Y^i_t,
M^i_t$ and a fixed realization $\pi$ of $\Pi_t$. We call
$\hat{g}^i_t(\cdot,\cdot,\pi)$ the \emph{partial control law} of controller $i$
at time $t$ for the given realization $\pi$ of the common information state
$\Pi_t$.

We now use Proposition \ref{prop:DP} to describe a dynamic programming decomposition of
the problem of finding optimal control strategies. This dynamic programming decomposition
allows us to evaluate optimal \emph{partial control laws} for each realization
$\pi$ of the common information state in a backward inductive manner. Recall
that $\mathcal{B}_t$ is the space of all possible realizations of $\Pi_t$ (see
\eqref{eq:beliefspace}) and $F(\mathcal{Y}^i_t \times
\mathcal{M}^i_t,\mathcal{U}^i_t)$ is the set of all functions from
$\mathcal{Y}^i_t \times \mathcal{M}^i_t$ to $\mathcal{U}^i_t$ (see Section
\ref{sec:notation}).

\begin{theorem}[Dynamic Programming Decomposition]\label{thm:seq_decomposition}
  Define the functions $V_{t} :
  \mathcal{B}_t \mapsto \reals$ , for $t=1,\dots,T$ as follows:
  \begin{equation}
    V_{T}(\pi) = \inf _{\{\tilde\gamma^i_T \in F(\mathcal{Y}^i_T \times \mathcal{M}^i_T,\mathcal{U}^i_T), 1\leq i \leq n\}} 
    \mathds{E}\{l(X_T,\tilde\gamma^1_T(Y^1_T,M^1_T),\ldots,\tilde\gamma^n_T(Y^n_T,M^n_T)) | 
    \Pi_T=\pi\},
  \end{equation}
  and for $1 \leq t \leq T-1$,
  \begin{multline}
    V_{t}(\pi) = \inf _{\{\tilde\gamma^i_t \in F(\mathcal{Y}^i_t \times \mathcal{M}^i_t,\mathcal{U}^i_t), 1\leq i\leq n\}} 
    \mathds{E}\big\{l(X_t,\tilde\gamma^1_t(Y^1_t,M^i_t),\ldots,\tilde\gamma^n_t(Y^n_t,M^n_t))  + \\ V_{t+1}(\eta_t(\pi,\tilde\gamma^1_t,\ldots,\tilde\gamma^n_t, \VEC Z_{t})) \,\big|\,
  \Pi_t=\pi \big\},
  \end{multline}
  where $\eta_t$ is a $\mathcal{B}_{t+1}$-valued function  defined in
  \eqref{eq:updatecoordinator} and Appendix \ref{sec:update_func}. 

  For $t=1,\dots,T$ and for each $\pi \in \mathcal{B}_t$, an optimal partial
  control law for controller $i$ is the minimizing choice of $\tilde\gamma^i$ in
  the definition of $V_t(\pi)$.  Let $\Psi_t(\pi)$ denote the $\arg\inf$ of
  the right hand side of $V_t(\pi)$, and $\Psi^i_t$ denote its $i$-th
  component. Then, an optimal control stategy is given by:
  \begin{equation}
    \hat g^i_t(\cdot, \cdot, \pi) = \Psi^i_t(\pi).
  \end{equation}
\end{theorem}

\subsection{Comparison with Person by Person and Designer Approaches}

The common information based approach adopted above differs from the
person-by-person approach and the designer's approach mentioned in the
introduction. In particular, the structural result of
Theorem~\ref{thm:str_result}  cannot be found by the person-by- person approach.
If we fix strategies of all but the $i$th controller to an arbitrary choice,
then it is \emph{not necessarily optimal} for controller $i$ to use a strategy
of the form in Theorem~\ref{thm:str_result}. This is because if controller $j$'s
strategy uses the entire common information $C_t$, then controller $i$, in
general, would need to consider the entire common information to better predict
controller $j$'s actions and hence controller $i$'s optimal choice of action may
too depend on the entire common information. The use of common information based
approach allowed us to prove that \emph{all controllers can jointly use
strategies of the form in Theorem~\ref{thm:str_result} without loss of
optimality}. 

The dynamic programming decomposition of Theorem~\ref{thm:seq_decomposition} is simpler
than any dynamic programming decomposition obtained using the designer's approach. As
described earlier, the designer's approach models the decentralized control
problem as an \emph{open-loop} centralized planning problem in which a designer
at each stage chooses control laws $g^i_t$ that map $(Y^i_t, M^i_t, C_t)$ to
$U^i_t$, $i=1,\dots,n$. On the other hand, the common-information approach
developed in this paper models the decentralized control problem as a
\emph{closed-loop} centralized planning problem in which a coordinator at each
stage chooses the \emph{partial} control laws $\gamma^i_t$ that map $(Y^i_t,
M^i_t)$ to $U^i_t$, $i=1,\dots,n$. The space of partial control laws is always
smaller than the space of full control laws; if the common information is
non-empty, then they are strictly smaller. Thus, the dynamic programming decomposition of
Theorem~\ref{thm:seq_decomposition} is simpler than that obtained by the
designer's approach. This simplification is best illustrated by the example of
Section~\ref{ex:identical} where all controllers receive a common observation
$Y^{com}_t$. For this example, we show that our information state (and hence our
dynamic program) reduce to $\PR(X_t|Y^{com}_{1:t})$, which is identical to the
information state of centralized stochastic control. In contrast,
the information state $\PR(X_t, Y^{com}_{1:t})$ obtained by the designer's
approach is much more complicated.

\subsection{Special Cases: The Results} \label{sec:special_results}

In Section \ref{sec:specialcases}, we described several models of decentralized
control problems that are special cases of the model described in
Section~\ref{sec:model}. In this section, we state the results of
Theorems~\ref{thm:str_result} and \ref{thm:seq_decomposition} for these models.

\subsubsection{Delayed Sharing Information Structure} \label{sec:ex1-result}

\begin{corollary}\label{cor:delayed-sharing}
  In the delayed sharing information structure of section~\ref{sec:ex1},  there exist optimal control strategies of the form
  \begin{equation}
    U^i_t = \hat{g}^i_t(Y^i_{t-s+1:t}, U^i_{t-s+1:t-1}, \Pi_t), \quad i= 1,2,\ldots,n,
  \end{equation}
  where     
  \begin{equation} 
    \Pi_t :=    \mathds{P}^{\hat{g}^{1:n}_{1:t-1}}(X_{t},
    \VEC Y_{t-s+1:t}, \VEC U_{t-s+1:t-1}| C_t).
  \end{equation}
  Moreover, optimal control strategies can be obtained by a dynamic program
  similar to that of Theorem~\ref{thm:seq_decomposition}.
\end{corollary}
The above result is analogous to the result in \cite{NMT:2011}.

\subsubsection{Delayed State Sharing Information Structure} \label{sec:ex2-result}
\begin{corollary}\label{cor:delayed-state}
  In the delayed state sharing information structure of section~\ref{sec:ex2},
  there exist optimal control strategies of the form
  \begin{equation} 
    U^i_t = \hat{g}^i_t(X^i_{t-s+1:t}, U^i_{t-s+1:t-1}, \Pi_t), \quad i= 1,2,\ldots,n,
  \end{equation}
  where     
  \begin{equation} 
    \Pi_t :=    \mathds{P}^{\hat{g}^{1:n}_{1:t-1}}( X_{t-s+1:t}, \VEC U_{t-s+1:t-1}| C_t).
  \end{equation}
  Moreover, optimal control strategies can be obtained by a dynamic program
  similar to that of Theorem~\ref{thm:seq_decomposition}.
\end{corollary}
The above result is analogous to the result in \cite{NMT:2011}.

\subsubsection{Periodic Sharing Information Structure} \label{sec:ex3-result}
\begin{corollary}\label{cor:periodic-sharing}
  In the periodic sharing information structure of section~\ref{sec:ex3},  there
  exist optimal control strategies of the form
  \begin{equation}
    U^i_t = \hat{g}^i_t(Y^i_{ks+1:t}, U^i_{ks+1:t-1}, \Pi_t), \quad i=
    1,2,\ldots,n, \quad ks < t \leq (k+1)s,
  \end{equation}
  where 
  \begin{equation} 
    \Pi_t :=    \mathds{P}^{\hat{g}^{1:n}_{1:t-1}}(X_{t},
    \VEC Y_{ks+1:t}, \VEC U_{ks+1:t-1}| C_t),  ~~ks < t \leq (k+1)s.
  \end{equation}
  Moreover, optimal control strategies can be obtained by a dynamic program
  similar to that of Theorem~\ref{thm:seq_decomposition}.
\end{corollary}
The above result gives a finer dynamic programming decomposition
that~\cite{OoiVerboutLudwigWornell:1997}.
In~\cite{OoiVerboutLudwigWornell:1997}, the dynamic programming decomposition
is only carried out at the times of information sharing, $t=ks$, $s=1,2,\dots$;
and at each step the partial control laws until the next sharing instant are
chosen. In contrast, in the above dynamic program, the partial control laws of
each step are chosen sequentially.

\subsubsection{Control Sharing Information Structure} \label{sec:ex4-result}
\begin{corollary}\label{cor:control-sharing}
  In the control sharing information structure of section~\ref{sec:ex4},  there
  exist optimal control strategies of the form
  \begin{equation}
    U^i_t = \hat{g}^i_t(Y^i_{1:t}, \Pi_t), \quad i= 1,2,\ldots,n,
  \end{equation}
  where 
  \begin{equation} 
    \Pi_t :=    \mathds{P}^{\hat{g}^{1:n}_{1:t-1}}(X_{t},
    \VEC Y_{1:t}| C_t).
  \end{equation}
  Moreover, optimal control strategies can be obtained by a dynamic program
  similar to that of Theorem~\ref{thm:seq_decomposition}.
\end{corollary}

\subsubsection{No Shared Memory with or without finite local memory} \label{sec:ex5-result}
\begin{corollary} \label{cor:no-memory}
  In the information structure of Section~\ref{sec:ex5}, there exist optimal
  control strategies of the form
  \begin{equation}\label{eq:no-memory-structure}
    U^i_t = \hat g^i_t(Y^i_t, M^i_t, \Pi_t)
  \end{equation}
  where
  \begin{equation}
    \Pi_t =   \mathds{P}^{\hat{g}^{1:n}_{1:t-1}}(X_{t}, \VEC Y_t,\VEC M_t)
  \end{equation}
  Moreover, optimal control strategies can be obtained by a dynamic program
  similar to that of Theorem~\ref{thm:seq_decomposition}.
\end{corollary}

Note that, since the common information is empty, the common information state $\Pi_t$
is now an \emph{unconditional} probability. In particular, 
$\Pi_t$ is a constant random variable and takes a fixed value  that
depends only on the choice of past control laws. Therefore, we can define an
appropriate control law $\tilde g^i_t$ such that
$\hat{g}^i_t(Y^i_t,M^i_t,\Pi_t) = \tilde g^i_t(Y^i_t,M^i_t)$, with
probability~$1$. Hence, the structural result of~\eqref{eq:no-memory-structure}
may be simplified to
\[ U^i_t = \hat{g}^i_t(Y^i_t,M^i_t,\Pi_t) = \tilde g^i_t(Y^i_t,M^i_t).\]
This result is redundant since all control laws are of the above form.
Nonetheless, Corollary~\ref{cor:no-memory} gives a procedure of finding such
control laws using the dynamic program of Theorem~\ref{thm:seq_decomposition}. 

The above result is similar to the results in \cite{Sandell:phd} for the case of
one controller with finite memory and to those in \cite{Mahajan_thesis} for the
case of two controllers with finite memories.
     
\section{Simplifications and  Generalizations}\label{sec:generalization}

\subsection{Simplification of the Common Information State}

Theorems~\ref{thm:str_result} and \ref{thm:seq_decomposition} identify the
conditional probability distribution on $(X_t,\VEC Y_t,\VEC M_t)$ given $C_t$ as
the common information state for our problem. In the following lemma, we make
the simple observation that in our model the  conditional distribution on $(X_t,
\VEC Y_t, \VEC M_t)$ given $C_t$ is completely determined by the conditional
distribution on $(X_t,\VEC M_t)$ given $C_t$. 
\begin{lemma}\label{lemma:newinfostate}
  For any choice of control laws $\hat{g}^{1:n}_{1:t-1}$, define the conditional
  distribution on $X_t, \VEC M_t$ given $C_t$ as  
  \[ \Pi^{new}_t(x, \VEC m) :=  \mathds{P}^{\hat{g}^{1:n}_{1:t-1}}(X_{t}=x,\VEC
  M_t=\VEC m| C_t),\] for all possible realizations $(x, \VEC m)$ of ($X_t, \VEC
  M_t$). Also define $\mathcal{B}^{new}_t:= \Delta(\mathcal{X}_t \times
  \mathcal{M}^i_t\times\ldots\times\mathcal{M}^n_t)$. Then, 
  \begin{equation}\label{eq:newinfostate1}
    \Pi^{new}_t(x,\VEC m) = \sum_{\VEC y}\Pi_t(x,\VEC y, \VEC m).
  \end{equation}
  Therefore, $\Pi^{new}_t = \chi_t(\Pi_t)$, where each component of the
  $\mathcal{B}^{new}_t$- valued function $\chi_t$ is determined by the right
  hand side of \eqref{eq:newinfostate1}. Also,
  \begin{equation}\label{eq:newinfostate2}
    \Pi_t(x,\VEC y,\VEC m) = \Pi^{new}_t(x, \VEC m)\mathds{P}(\VEC Y_t=\VEC y|X_t =x),
  \end{equation} 
  where the second term on right hand side of \eqref{eq:newinfostate2} is
  determined by the fixed distribution of the observations noises. Therefore,
  $\Pi_t = \zeta_t(\Pi^{new}_t)$, where each component of the $\mathcal{B}_t$-
  valued function $\zeta_t$ is determined by the right hand side of
  \eqref{eq:newinfostate2}. 
\end{lemma}
Lemma \ref{lemma:newinfostate} implies that the results of Theorems~\ref{thm:str_result} and \ref{thm:seq_decomposition} can be written in terms of $\Pi^{new}_t$.

\begin{theorem}[Alternative Common Information State]\label{thm:newinfostate}
  In Problem~\ref{prob:finite},  there exist optimal  control strategies of the
  form
  \begin{equation} \label{eq:our_newresult}
    U^i_t = \hat{\hat{g}}^i_t(Y^i_t, M^i_t, \Pi^{new}_t), \quad i= 1,2,\ldots,n,
  \end{equation}
  where     
  \begin{equation} 
    \Pi^{new}_t :=    \mathds{P}^{\hat{\hat{g}}^{1:n}_{1:t-1}}(X_{t},\VEC M_t | C_t).
  \end{equation}
  Further, define the functions 
  $V^{new}_t:  \mathcal{B}^{new}_t \mapsto \reals$ , for $t=1,\dots,T$ as follows:
  \begin{equation}
    V^{new}_{T}(\pi^{new}) = \inf _{\{\tilde\gamma^i_T \in F(\mathcal{Y}^i_T \times \mathcal{M}^i_T,\mathcal{U}^i_T), 1\leq i \leq n\}} 
    \mathds{E}\{l(X_T,\tilde\gamma^1_T(Y^1_T,M^1_T),\ldots,\tilde\gamma^n_T(Y^n_T,M^n_T)) | 
    \Pi_T=\zeta_T(\pi^{new})\},
  \end{equation}
  and for $1 \leq t \leq T-1$,
  \begin{multline}
    V^{new}_{t}(\pi^{new}) = \inf _{\{\tilde\gamma^i_t \in F(\mathcal{Y}^i_t \times \mathcal{M}^i_t,\mathcal{U}^i_t), 1\leq i\leq n\}} 
    \mathds{E}\big\{l(X_t,\tilde\gamma^1_t(Y^1_t,M^i_t),\ldots,\tilde\gamma^n_t(Y^n_t,M^n_t))  + \\ V^{new}_{t+1}(\chi_t(\eta_t(\Pi_t,\tilde\gamma^1_t,\ldots,\tilde\gamma^n_t, \VEC Z_{t}))) \,\big|\,
    \Pi_t=\zeta_t(\pi^{new}) \big\},
  \end{multline}
  where $\zeta_t, \chi_t$ are defined in Lemma~\ref{lemma:newinfostate}, and
  $\eta_t$ is   defined in \eqref{eq:updatecoordinator} and Appendix
  \ref{sec:update_func}. 

  For $1 \leq t \leq T$ and for each $\pi^{new}$, an optimal partial control law
  for controller $i$ is the minimizing choice of $\tilde\gamma^i$ in the
  definition of $V^{new}_t(\pi^{new})$.
\end{theorem}

\begin{proof}
  For any $\pi^{new} \in \mathcal{B}^{new}_t$ and any $\pi \in \mathcal{B}_t$,
  it is straightforward to establish using a backward induction argument that
  $V^{new}_t(\pi^{new}) = V_t(\zeta_t(\pi^{new}))$ and $V_t(\pi) =
  V^{new}_t(\chi_t(\pi))$, where $V_t(\cdot)$ is the value function from the
  dynamic program in Theorem~\ref{thm:seq_decomposition}. The optimality of the
  new dynamic program then follows from the optimality of the dynamic program in
  Theorem~\ref{thm:seq_decomposition}.  
\end{proof}

The result of Theorem~\ref{thm:newinfostate} is conceptually the same as the
results in Theorems~\ref{thm:str_result} and \ref{thm:seq_decomposition}.
Theorem~\ref{thm:newinfostate} implies that the Corollaries of
Section~\ref{sec:special_results} can be restated in terms of new information
states by simply removing $Y_t$ from the definition of original information
states. For example, the result of Corollary~\ref{cor:delayed-sharing} for
delayed sharing information structure is also
true when $\Pi_t$ is replaced by
\begin{equation} \label{eq:delay-sharing}
  \Pi^{new}_t := \mathds{P}^{\hat g^{1:n}_{1:t-1}}(X_t,
      \mathbf Y_{t-s+1:t-1}, \mathbf U_{t-s+1:t-1} | C_t). 
\end{equation}
      This result is simpler than that of~\cite[Theorem~2]{NMT:2011}.

\subsection{Generalization of the  Model}
The methodology described in Section~\ref{sec:proof} relies on the fact that the
shared memory is \emph{common information} among all controllers. Since the
coordinator in the coordinated system knows only the common information, any
coordination strategy can be mapped to an equivalent control strategy in the
basic model (see Stage 4 of Section~\ref{sec:proof}). In some cases, in addition
to the shared memory, the current observation (or if the current observation is
a vector, some components of it) may also be commonly available to all
controllers. The general methodology of Section 2 can be easily modified to
include such cases as well.

Consider the model of Section~\ref{sec:model} with the following modifications:
\begin{enumerate}
  \item  In addition to their current local observation, all controllers have a
    \emph{common observation} at time $t$. 
    \begin{equation}  \label{eq:common_observation}
      Y^{com}_t = h^{com}_t(X_t, V_t)
    \end{equation}
    where $\{V_t, t=1,\dots,T\}$ is a sequence of i.i.d\@. random variables with
    probability distribution $Q_V$ which is independent of all other primitive
    random variables.

  \item  The shared memory $C_t$ at time $t$ is a subset of
    $\{Y^{com}_{1:t-1},\VEC Y_{1:t-1}, \VEC U_{1:t-1}\}$.

  \item Each controller selects its action using a control law of the form
    \begin{equation}\label{eq:newcontrol}
      U^i_t = g^i_t(Y^i_t, M^i_t, C_t,Y^{com}_t).
    \end{equation}

  \item After taking the control action at time $t$, controller~$i$ sends a
    subset $Z^i_t$ of  $\{M^i_t, Y^i_t, U^i_t, Y^{com}_t\}$ that necessarily
    includes $Y^{com}_t$. That is,
    \[Y^{com}_t \subset Z^i_t \subset \{M^i_t, Y^i_t, U^i_t, Y^{com}_t\}.  \]
    This implies that the history of common observations is necessarily a part
    of the shared memory, that is, $Y^{com}_{1:t-1} \subset C_t$.  
\end{enumerate}
The rest of the model is same as in Section~\ref{sec:model}. In particular, the
local memory update satisfies \eqref{eq:local update}, so the local memory and
shared memory at time $t+1$ don't overlap. The instantaneous cost is given by
$\COST(X_t, U_t)$ and the objective is to minimize an expected total cost given
by \eqref{eq:cost}.

The arguments of Section~\ref{sec:proof} are also valid for this model. The
observation process in Lemma~\ref{lemma:state} is now defined as $R_{t+1} =
\{\VEC Z_t, Y^{com}_{t+1}\}$. The analysis of Section~\ref{sec:proof} leads to
structural results and dynamic programming decompositions analogous to
Theorems~\ref{thm:str_result} and \ref{thm:seq_decomposition} with $\Pi_t$ now
defined as
\begin{equation} \label{eq:define_newpi}
  \Pi_t :=    \mathds{P}^{g^{1:n}_{1:t-1}}(X_{t}, \VEC Y_t,\VEC M_t | C_t,Y^{com}_t).
\end{equation}
Using an argument similar to Lemma~\ref{lemma:newinfostate}, we can show that
the result of Theorem~\ref{thm:newinfostate} is true for the above model with
$\Pi^{new}_t$ defined as
\begin{equation}
  \Pi^{new}_t :=
  \mathds{P}^{\hat g^{1:n}_{1:t-1}}(X_t, \VEC M_t | C_t, Y^{com}_t). 
\end{equation}

\subsection{Examples of the Generalized Model}\label{sec:generalized_examples}
\subsubsection{Controllers with Identical Information} \label{ex:identical}
Consider the following special case of the above generalized model.
\begin{enumerate}
  \item All controllers only make the common observation $Y^{com}_t$;
    controllers have no local observation or local memory. 
  \item The shared memory at time $t$ is $C_t = Y^{com}_{1:t-1}$. Thus, at time
    $t$, all controllers have identical information given as $\{C_t,Y^{com}_t\}
    = Y^{com}_{1:t}$. 
  \item After taking the action at time $t$, each controller sends $Z^i_t =
    Y^{com}_t$ to the shared memory.
\end{enumerate}

Recall that the coordinator's prescription $\Gamma^i_t$ in
Section~\ref{sec:proof} are chosen from the set of functions from
$\mathcal{Y}^i_t \times \mathcal{M}^i_t$ to $\mathcal{U}^i_t$. Since, in this
case $\mathcal{Y}^i_t = \mathcal{M}^i_t = \emptyset$, we interpret the
coordinator's prescription as prescribed actions. That is, $\Gamma^i_t \equiv
U^i_t$. With this interpretation, the common information state becomes
\begin{equation} \label{eq:pomdp_info_new}
  \Pi_t := \mathds{P}^{ g^{1:n}_{1:t-1}}(X_t|Y^{com}_{1:t}) 
\end{equation}
and the dynamic program of Theorem~\ref{thm:seq_decomposition} becomes
\begin{equation}
    V_{T}(\pi) = \inf _{\{u^i_T  \in \mathcal{U}^i_T), 1\leq i \leq n\}} 
    \mathds{E}\{l(X_T,u^1_T,\ldots,u^n_T) | 
    \Pi_T=\pi\},
\end{equation}
and for $1 \leq t \leq T-1$,
\begin{multline}
    V_{t}(\pi) = \inf _{\{u^i_t  \in \mathcal{U}^i_T), 1\leq i \leq n\}} 
    \mathds{E}\big\{l(X_t,u^1_t,\ldots,u^n_t)+ V_{t+1}(\eta_t(\pi,u^1_t,\ldots,u^n_t, Y^{com}_{t+1})) \,\big|\,
    \Pi_t=\pi \big\}.
\end{multline}
Since all the controllers have identical information, the above results
correspond to the centralized dynamic program of Theorem~\ref{thm:pomdp} with a
single controller choosing all the actions. 


\subsubsection {Coupled subsystems with control sharing information structure}

Consider the following special case of the above generalized model.
\begin{enumerate}
  \item The state of the system at time $t$ is a $(n+1)$-dimensional vector
    $X_t = (X^1_t, X^2_t, \dots, X^n_t, X^*_t)$, where $X^i_t$, $i=1,\dots,n$
    corresponds to the local state of subsystem~$i$, and $X^*_t$ is a global state of
    the system.
  \item The state update function is such that the global state evolves
    according to
    \[ X^*_{t+1} = f^*_t(X^*_t, \VEC U_t, N^0_t), \]
    while the local state of subsystem~$i$ evolves according to
    \[ X^i_{t+1} = f^i_t(X^i_t, X^*_t, \VEC U_t,N^i_t), \]
    where $\{N^0_t, t=1,\dots T\}, \ldots, \{N^n_t, t=1,\dots T\}$ are mutually independent i.i.d noise
      processes that are independent of the initial state, $X_1 = (X^1_1, X^2_1, \dots, X^n_1, X^*_1)$.
  \item At time $t$, the common observation of all controllers is given by $Y^{com}_t = X^*_t$. 
  \item At time $t$, the local observation of controller~$i$ is given by $Y^i_t
    = X^i_t$, $i=1,\dots,n$. 
  \item The shared memory at time $t$ is $C_t = \{X^*_{1:t-1}, \VEC
    U_{1:t-1}\}$. At each time $t$, after taking the action $U^i_t$,
    controller~$i$ sends $Z^i_t = \{X^*_t,U^i_t\}$ to the shared memory. 
\end{enumerate}

The above special case corresponds to the model of coupled subsystems with
control sharing considered in \cite{Mahajan:2011}, where several applications of
this model are also presented. It is shown in \cite{Mahajan:2011} that there is
no loss of optimality in restricting attention to controllers with no local
memory, i.e., $M_t = \emptyset$. With this additional restriction, 
the result of Theorems 1 and 2 apply for this model with $\Pi_t$
defined as
\[\Pi_t := \mathds{P}^{g^{1:n}_{1:t-1}}(X^{*}_t,X^1_t,\ldots,X^n_t|X^*_{1:t},\VEC U_{1:t-1}). \]
Note that $\Pi_t$ can be evaluated from $X^{*}_t$ and
$\mathds{P}^{g^{1:n}_{1:t-1}}(X^1_t,\ldots,X^n_t|X^*_{1:t},\VEC U_{1:t-1})$. It
is  shown in \cite{Mahajan:2011} that $X^1_t,X^2_t,\ldots,X^n_t$ are
conditionally independent given  $X^*_{1:t},\VEC U_{1:t-1}$, hence the joint
distribution $\mathds{P}^{g^{1:n}_{1:t-1}}(X^1_t,\ldots,X^n_t|X^*_{1:t},\VEC
U_{1:t-1})$ is a product of its marginal distributions.

\subsubsection {Broadcast information structure}
Consider the following special case of the above generalized model.

\begin{enumerate}
  \item The state of the system at time $t$ is a $n$-dimensional vector
    $X_t = (X^1_t, X^2_t, \dots, X^n_t)$, where $X^i_t$, $i=1,\dots,n$
    corresponds to the local state of subsystem~$i$. The first component $i=1$
    is special and called the \emph{central node}. Other components,
    $i=2,\dots,n$, are called \emph{peripheral nodes}.
  \item 
    The state update function is such that the state of the central node evolves according to
    \[ X^1_{t+1} = f^1_t(X^1_t, U^1_t, N^1_t) \]
    while the state of the peripheral nodes evolves according to
    \[ X^i_{t+1} = f^i_t(X^i_t, X^1_t, U^i_t,U^1_t, N^i_t) \]
    where $\{N^i_t, i=1,2,\ldots n; t=1,\dots\}$ are noise processes that are independent across
    time and independent of each other.
  \item At time $t$, the common observation of all controllers is given by $Y^{com}_t = X^1_t$. 
  \item At time $t$, the local observation of controller~$i$, $i > 2$, is given by $Y^i_t
    = X^i_t$. Controller~1 does not have any local observations.
  \item No controller sends any additional data to the shared memory.
    Thus, the shared memory consists of just the history of common observations,
    \emph{i.e.}, $C_t = Y^{com}_{1:t-1} = X^1_{1:t-1}$.
\end{enumerate}

The above special case corresponds to the model of decentralized systems with
broadcast structure considered in \cite{WuLall:2010}. 
It is shown in \cite{WuLall:2010} that there is
no loss of optimality in restricting attention to controllers with no local
memory, i.e., $M_t = \emptyset$. With this additional restriction, 
the result of Theorems 1 and 2 apply for this model with $\Pi_t$
defined as
\[\Pi_t := \mathds{P}^{g^{1:n}_{1:t-1}}(X^1_t,\ldots,X^n_t|X^1_{1:t}). \]
Note that $\Pi_t$ can be evaluated from $X^{1}_t$ and
$\mathds{P}^{g^{1:n}_{1:t-1}}(X^2_t,\ldots,X^n_t|X^1_{1:t})$. It is  shown in
\cite{WuLall:2010} that $X^2_t,\ldots,X^n_t$ are conditionally independent given
$X^1_{1:t}$, hence the joint distribution
$\mathds{P}^{g^{1:n}_{1:t-1}}(X^2_t,\ldots,X^n_t|X^1_{1:t})$ is a product of its
marginal distributions.
 
\section{Extension to infinite horizon} \label{sec:infinite}

In this Section, we consider the basic model of Section~\ref{sec:model} with an
infinite time horizon. Assume that 
\begin{enumerate}
  \item[(i)] The state of the system, the observations and the control actions take
    value in time-invariant sets $\mathcal{X}, \mathcal{Y}^i, \mathcal{U}^i$,
    respectively.
  \item[(ii)] The local memories $M^i_t$ and the updates to the shared memory $Z^i_t$
    take values in time-invariant sets $\mathcal{M}^i$ and $\mathcal{Z}^i$
    respectively.
  \item[(iii)] The dynamics of the system (equation~\eqref{eq:state}) and the
    observation model (equation~\eqref{eq:observation}) are time-homogeneous.
    That is, the functions $f_t$ and $h_t$ in equations \eqref{eq:state} and
    \eqref{eq:observation} do not vary with time.
\end{enumerate}

Let the cost of using a strategy $\VEC g^{1:n}$ be defined as
\begin{equation}
  \label{eq:infinite_cost}
  J(\VEC g^{1:n}) \DEFINED \EXP^{\VEC g^{1:n}}\Big[ 
  \sum_{t=1}^\infty \beta^{t-1} \COST(X_t, \VEC U_t) \Big],
\end{equation}
where $\beta \in [0,1)$ is a discount factor. We can follow the arguments of
Section~\ref{sec:proof} to formulate the problem of the coordinated system with
an infinite time horizon. As in Section~\ref{sec:proof}, the coordinated system
is equivalent to a POMDP. The time-homogeneous nature of the coordinated system
and its equivalence to a POMDP allows us to use known POMDP results (see
\cite{Whittle:1983}) to conclude the following theorem for the infinite time
horizon problem.

\begin{theorem}\label{thm:inf_horizon}
  Consider Problem~\ref{prob:finite} with infinite time horizon and the
  objective of minimizing the expected cost given by
  equation~\eqref{eq:infinite_cost}. Then, there exists an  optimal
  time-invariant control strategy  of the form:
  \begin{equation} \label{eq:inf_result}
    U^i_t = g^i(Y^i_t, M^i_t, \Pi_t), \quad i= 1,2,\ldots,n,
  \end{equation}
  Furthermore, consider the fixed point equation, 
  \begin{multline} 
    V(\pi) = \inf _{\{\tilde\gamma^i \in F(\mathcal{Y}^i \times \mathcal{M}^i,\mathcal{U}^i), 1\leq i\leq n\}} 
    \mathds{E}\big\{l(X_t,\tilde\gamma^1_t(Y^1_t,M^i_t),\ldots,\tilde\gamma^n_t(Y^n_t,M^n_t))+ \\ \beta V(\eta_{t}(\pi,\tilde\gamma^1,\ldots,\tilde\gamma^n,\VEC{Z}_{t})) \,\big|\, 
  \Pi_t=\pi \big\}. \label{eq:fixedpointeq}
  \end{multline}
  Then, for any realization $\pi$ of $\Pi_t$, the optimal partial control laws
  are the choices of $\gamma^i$ that achieve the infimum in the right hand side
  of \eqref{eq:fixedpointeq}.
\end{theorem}

All the special cases of our information structure considered in
Sections~\ref{sec:specialcases} and \ref{sec:generalized_examples} can be
extended to infinite horizon problems if the state, observation and actions
spaces are time-invariant and the systems dynamics and observation equations are
time homogeneous. The only exception is the control sharing information
structure of section~\ref{sec:ex4} where the local memory takes values in sets
that are increasing with time.

\emph{The Case of No Shared Memory:}
As discussed in Section \ref{sec:special_results}, if the shared memory is
always empty then the common information state defined in
Theorem~\ref{thm:str_result} is the \emph{unconditional} probability  $\Pi_t =
\mathds{P}^{g^{1:n}_{1:t-1}}(X_{t}, \VEC Y_t,\VEC M_t)$. In particular,  $\Pi_t$
is a random variable that takes a fixed (constant) value which depends only on
the choice of past control laws. Therefore, for any function $g^i_t$ of
$Y^i_t,M^i_t,\Pi_t$, there exists a function $\tilde g^i_t$ of $Y^i_t,M^i_t$
such that  $\tilde g^i_t(Y^i_t,M^i_t) =  g^i_t(Y^i_t,M^i_t,\Pi_t)$ with
probability $1$. While Theorem \ref{thm:inf_horizon} establishes optimality of a
time-invariant $g^i_t$,  such time-invariance may not hold for the corresponding
$\tilde g^i_t$. Similar observations were reported in
\cite{MahajanTeneketzis:2010}. 


\section{Discussion and Conclusions} \label{sec:conclusion}

In centralized stochastic control, the controller's belief on the current state
of the system plays a fundamental role  for predicting future costs. If the
control strategy for the future is fixed as a function of future beliefs, then
the current belief is a sufficient statistic for future costs under any choice
of current action. Hence, the optimal action at the current time is only a
function of current belief on the state. In decentralized problems where
different controllers have different information, using a controller's belief on
the state of the system presents two main difficulties: (i) Since the costs
depend both on system state as well as other controllers' actions  any
prediction of future costs must involve a belief on system state as well as some
means of predicting  other controllers' actions. (ii) Secondly, since different
controllers have different information, the beliefs formed by each controller
and their predictions of future costs cannot be expected to be consistent. 

The approach we adopted in this paper tries to address these difficulties by
using the fact that sharing of data among controllers creates common knowledge
among the controllers. Beliefs based on this common knowledge are necessarily
consistent among all controllers and can serve as a consistent sufficient
statistic. Moreover, while controllers cannot accurately predict each other's
control actions, they can know, for the observed realization of common
information, the exact mapping used by each controller to map its local
information  to control action. These considerations suggest that common
information based beliefs and partial control laws should play an important role
in a general theory of decentralized stochastic control problems. The use of a
fictitious coordinator allows us to make these considerations mathematically
precise. Indeed,  the coordinator's beliefs are based on common information and
the coordinator's decision are the partial control laws. The results of the
paper then follow by observing that the coordinator's problem can be viewed as a
POMDP by identifying a new state that includes both the state of the dynamic
system as well as the local information of the controllers.

The specific model of shared and local memory update that we assumed is crucial
for connecting the coordinator's problem to POMDPs and centralized stochastic
control. A key assumption in centralized stochastic control is perfect recall,
that is, the information obtained at any time is remembered at all future times.
This is essential for the update of the beliefs in POMDPs. Our assumption that
the shared memory is increasing in time ensures that the perfect recall property
is true for the coordinator's problem. If the shared memory did not have perfect
recall (that is, if some past contents were lost over time), then the update of
common information state in \eqref{eq:updatecoordinator} would not hold and the
results of Theorems~\ref{thm:str_result} and \ref{thm:seq_decomposition} would
not be true. 

Another key factor in our result is that $S_t :=\{X_t, \VEC Y_t, \VEC M_t\}$
serves as a state for the coordinator's problem. If the system state,
observations and local memories take value in a time-invariant space, we have a
state for the coordinator's problem which takes value in a time-invariant space.
Hence, the common information state is a belief on a time-invariant space. The
local memory update in \eqref{eq:local update} ensures that  $S_t$ is a state.
If local memory update depended on shared memory as well, that is, if
\eqref{eq:local update} were replaced by
\[  M^i_{t+1} \subset \{C_t,M^i_t,Y^i_t,U^i_t\},\]
then $S_t$ would no longer suffice as a state for the coordinator. In
particular, the state update equations in Lemma \ref{lemma:state} would no
longer hold. The only recourse then would be to include $C_t$ as a part of the
state which would necessarily mean that the state space keeps increasing with
time. This is undesirable not only because large state spaces imply increased
complexity, but the increasing size of state spaces also makes extensions of
finite horizon results to infinite horizon problems conceptually difficult.

The connection between the coordinator's problem and POMDPs can be used for
computational purposes as well. The dynamic  program of
Theorem~\ref{thm:seq_decomposition} is essentially a POMDP dynamic program. In
particular, just as in POMDP, the value-functions are piecewise linear and
concave in $\pi_t$. This characterization of value functions is utilized to find
computationally efficient algorithms for POMDPs. Such algorithmic solutions to
general POMDPs are well-studied and can be employed here. We refer the reader to
\cite{Zhang:2009} and references therein for a review of algorithms to solve
POMDPs.
  
While our results apply to a broad class of models, it would be worthwhile to
identify special cases where the specific model features can be exploited to
simplify our structural result. Examples of such  simplification appear in
\cite{Mahajan:2011, WuLall:2010}. A common theme in many centralized dynamic
programming solutions is to identify a key property of the value functions and
use it to characterize the optimal decisions. Since our results also provide a
dynamic program, an important avenue for future work would be to identify cases
where properties of value functions can be analyzed to deduce a solution or to
reduce the computational burden of finding the solution. 


Our approach in this paper illustrates that common information provides a common
conceptual framework for several decentralized stochastic control problems. In
our model, we explicitly included a shared memory which naturally served the
purpose of common information among the controllers. More generally, we can
\emph{define} common information for any sequential decision-making problem and
then address the problem from the perspective of a coordinator who knows the
common information. Such a common information based approach for general
sequential decision-making problems is presented in \cite{mythesis}.  


\section{Acknowledgments}

This work  was supported by  NSERC through the grant NSERC-RGPIN 402753-11, by
NSF through the grant CCF-1111061 and by  NASA through the grant NNX09AE91G. 

\bibliographystyle{IEEEtran}
\bibliography{IEEEabrv,myref,collection}

\appendices
\section{The Update Function $\eta_t$ of the Coordinator's Information State}\label{sec:update_func}

 Consider a realization
$c_{t+1}$ of the shared memory $C_{t+1}$ at time $t+1$. Let $(\VEC \gamma_{1:t})$ be the corresponding
realization of the coordinator's prescriptions until time $t$. We assume the
realization $(c_{t+1},\pi_{1:t},\VEC \gamma_{1:t})$ to be of non-zero probability. Then, the realization
$\pi_{t+1}$ of $\Pi_{t+1}$ is given by
\begin{equation}\label{eq:app-info-1}
  \pi_{t+1}(s) = \mathds{P}\{S_{t+1} = s | c_{t+1}, \VEC \gamma_{1:t}\}.
\end{equation}
Use Lemma~\ref{lemma:state} to simplify the above expression as
\begin{align} 
  \hskip 2em & \hskip -2em
  \sum_{s_t, w^0_t, \VEC w_{t+1}} 
  \IND_{s}( \tilde f_{t}(s_t, \VEC \gamma_t, w^0_t, \VEC w_{t+1} ))
   \cdot 
   \mathds{P}\{W^0_t = w^0_t\} \cdot \mathds{P}\{\VEC W_{t+1} = \VEC w_{t+1}\}
   \cdot 
   \mathds{P}\{S_t = s_t | c_{t+1}, \VEC \gamma_{1:t}\}.
\label{eq:app-info-2}
\end{align}
Since $c_{t+1} = (c_t, \VEC z_{t})$, write the last term
of~\eqref{eq:app-info-2} as
\begin{align}
  \mathds{P}\{S_t = s_t| c_{t}, \VEC z_{t}, \VEC \gamma_{1:t}\} = \frac
  {\mathds{P}\{S_t = s_t, \VEC Z_{t} = \VEC z_{t} | c_t, \VEC \gamma_{1:t}\}}
  {\sum_{s'} \mathds{P}\{S_t = s', \VEC Z_{t} = \VEC z_{t} |  c_t, \VEC \gamma_{1:t}\}}.
  \label{eq:app-info-3}
\end{align}

Use Lemma \ref{lemma:state} and the sequential order in which
the system variables are generated to write the numerator as
\begin{align}
  \hskip 2em & \hskip -2em
  \mathds{P}\{S_t = s_t, \VEC Z_{t} = \VEC z_{t} | c_t, \VEC \gamma_{1:t}\}
  = \IND_{\tilde h_t(s_t, \VEC \gamma_t)}(\VEC z_{t}) 
     \cdot
     \mathds{P}\{S_t = s_t| c_t, \VEC \gamma_{1:t}\} \label{eq:new_eq_in_app_A} \\
  &= \IND_{\tilde h_t(s_t, \VEC \gamma_t)}(\VEC z_{t})\cdot\pi_t(s_t).
  \label{eq:app-info-4}
\end{align}
where we dropped $\VEC \gamma_t$ from conditioning
in (\ref{eq:new_eq_in_app_A}) since under the given coordinator's
strategy, it is a function of the rest of the terms in the
conditioning. Substitute~\eqref{eq:app-info-4},
\eqref{eq:app-info-3}, and \eqref{eq:app-info-2}
into~\eqref{eq:app-info-1}, to get
\begin{equation*}
  \pi_{t+1}(s) = \eta_{t}^{s}(\pi_t, \VEC \gamma_t, \VEC \VEC z_{t}),
\end{equation*}
where $\eta_{t}^{s}(\cdot)$ is given by~\eqref{eq:app-info-1}, \eqref{eq:app-info-2},
\eqref{eq:app-info-3}, and~\eqref{eq:app-info-4}. $\eta_t(\cdot)$ is the vector  $(\eta^s_t(\cdot))_{s \in \mathcal{S}}$.

\section{Proof of Proposition \ref{prop:equiv}} \label{sec:equiv_proof}

(a) For any given control strategy $\VEC g^{1:n}$ in the basic model, define a coordinated strategy $\VEC d$ for the coordinated system as
\begin{equation} \label{eq:equiv1}
 d_t(C_t) = 
          \big(g^1_t( \cdot, \cdot, C_t), \dots, g^n( \cdot, \cdot, C_t) \big).
\end{equation}
Consider Problems~\ref{prob:finite} and~\ref{prob:coordinator}. Use
control strategy $\VEC g^{1:n}$ in Problem~\ref{prob:finite} and coordination
strategy $\VEC d$ given by~\eqref{eq:equiv1} in
Problem~\ref{prob:coordinator}. Fix a specific realization of the
primitive random variables $\{X_1,W^j_t,t=1,\dots,T, j =0, 1,\dots, n\}$ in the two problems. Equation~\eqref{eq:observation} implies that the realization of $\VEC Y_1$ will be the same in the two problems. Then, the choice of $\VEC d$ according
to~\eqref{eq:equiv1} implies that the realization of the control actions $\VEC U_1$ will be the same in the two problems. This implies that the realization of the next state $X_2$ and the memories $\VEC M_2$, $C_2$ will be the same in the two problems. Proceeding in a similar manner, it is clear that the choice of $\VEC d$ according
to~\eqref{eq:equiv1} implies that the realization of the state $\{X_t;
\allowbreak t=1,\dots,T\}$, the observations $\{\VEC Y_t;
\allowbreak t=1,\dots,T\}$,  the control actions $\{\VEC U_t;
\allowbreak t=1,\dots,T\}$ and the memories $\{\VEC M_t;
\allowbreak t=1,\dots,T\}$ and $\{C_t;
\allowbreak t=1,\dots,T\}$ are all identical in Problem~\ref{prob:finite}
and~\ref{prob:coordinator}. Thus, the total expected cost under
$\VEC g^{1:n}$ in Problem~\ref{prob:finite} is same as the total
expected cost under the coordination strategy given
by~\eqref{eq:equiv1} in Problem~\ref{prob:coordinator}. That is,
$ J(\VEC g^{1:n}) = \hat J(\VEC d)$.

(b) The second part of Proposition \ref{prop:equiv} follows from similar arguments as above. 

\section{Equivalence between the model of this paper and the model
of~\cite{MahajanNayyarTeneketzis:2008}}
\label{app:equiv}

We refer to the model of this paper as the PHS (partial history sharing) model
and the model of~\cite{MahajanNayyarTeneketzis:2008} as the CO (common
observation) model. First, we describe the CO model and then show the both models
are equivalent by showing that the PHS model is a special case of CO model and
vice versa.

\subsection*{The CO Model}
The following model was presented in~\cite{MahajanNayyarTeneketzis:2008}; we use
a slightly different notation so that the notation matches with that of our
paper.

Consider a system with $n$ controllers. Let $X_t$ denote the state of the
system, $Z_t$ denote the common observation of all controllers, $Y^i_t$ denote
the private observation of controller~$i$, $M^i_t$ the contents of the memory of
controller~$i$, and $U^i_t$ the control action of controller~$i$, $i=1,\dots,n$.

The system dynamics and observation equations are given by
\begin{align}
  X_{t+1} &= f_t(X_t, U^{1:n}_t, W^0_t), \\
  Y^i_t   &= h^i_t(X_t, U^{1:i-1}_t, W^i_t), \quad i=1,\dots,n,\\
  Z_t     &= c_t(X_t, U^{1:n}_{t-1}, Q_t),
\end{align}
where $\{X_1, Q_t, W^i_t, i=0,\dots,n, t=1,\dots,T\}$ are independent random
variables. 

At time~$t$, controller~$i$ generates a control action and updates its memory as
follows:
\begin{align}
  U^i_t &= g^i_t(Z_{1:t}, Y^i_t, M^i_{t-1}), \\
  M^i_t &= r^i_t(Z_{1:t}, Y^i_t, M^i_{t-1}).
\end{align}

At each time an instantaneous cost $l_t(X_t, U^{1:n}_t)$ is incurred. The system
objective is to choose a control strategy $g^{1:n}_{1:T}$ and a memory update
strategy $r^{1:n}_{1:T}$ to minimize a total expected cost.

\subsection*{The PHS model is a special case of CO model}

Consider the PHS model described in Sec~II-A of the paper and define
\begin{align*}
  \tilde X_t &= (X_t, Y^{1:n}_{t}, M^{1:n}_{t}, Z^{1:n}_{t-1}), \\
  \tilde U^i_t &= U^i_t, \quad i=1,\dots,n\\
  \tilde Y^i_t &= (Y^i_t, M^i_{t}), \quad i=1,\dots,n\\
  \tilde Z_t   &= Z^{1:n}_{t-1}, \\
  \tilde M^i_t &= \emptyset, \quad i=1,\dots,n.
\end{align*}
Define the cost function
\[ \tilde l_t(\tilde X_t, \tilde U^{1:n}_t) = l_t(X_t, U^{1:n}_t). \]

It is easy to verify that the model $(\tilde X_t, \tilde U^{1:n}_t, \tilde
Y^{1:n}_t, \tilde M^{1:n}_t, \tilde Z_t)$ defined above is a special case of CO model. 

\subsection*{The CO model is a special case of PHS model}

In the CO model, the local observations $Y^i_t$ of controller~$i$ depends on the
control action $U^{1:i-1}_t$. This feature is not present in PHS model.
Nonetheless, we can show that CO model is a special case of the PHS model by
splitting time and assuming that in the PHS model only one controller acts at
each time.

Define the following system variables for $\tau = 1, \dots, nT$. For ease of
notation, when $tn < \tau \le (t+1)n$, we will write $\tau$ as $tn + i$. Thus,
the system variables are defined for $t=1,\dots,T$ and $i=1,\dots,n$:
\begin{align*}
  \tilde X_{tn+1} &= (X_t, U^{1:n}_{t-1}, M^{1:n}_{t-1}), &
  \tilde X_{tn+i} &= (X_t, U^{i:i-1}_t, M^{1:i-1}_{t}, M^{i:n}_{t-1}), \quad i=2,\dots,n,
  \\
  \tilde Z_{tn+1} &= Z_t, &
  \tilde Z_{tn+i} &= \emptyset, \quad i=2,\dots,n,
\end{align*}
\begin{alignat*}{2}
  \tilde Y^i_{tn+j}    &= \begin{cases} (Y^1_t, M^1_{t-1}, Z_t), & \text{if $i=j=1$}, \\
                                        (Y^i_t,M^i_{t-1}) & \text{if $i=j\neq 1$}, \\
                                        \emptyset, &\text{otherwise};
                          \end{cases}  \quad &&i,j=1,\dots,n\\
  \tilde U^i_{tn+j}    &= \begin{cases} (U^i_t,M^i_t), & \text{if $i=j$}, \\
                                        \emptyset, &\text{otherwise};
                          \end{cases}\quad  &&i,j=1,\dots,n\\
  \tilde M^i_{tn+j} &= \emptyset, &&j=1,\dots,n.
\end{alignat*}

Define the cost function as:
\[ \tilde l_{tn+i}(\tilde X_{tn+i}, \tilde U^{1:n}_{tn+i})
  = \begin{cases}
      l_t(X_t, U^{1:n}_t), & \text{if $i=n$}, \\
      0, & \text{otherwise}.
    \end{cases}
\]

It is easy to verify that the model $(\tilde X_\tau, \tilde U^{1:n}_\tau, \tilde
Y^{1:n}_\tau, \tilde M^{1:n}_\tau, \tilde Z_\tau)$ defined above is a special case of PHS model.

\end{document}